\documentclass{article}

\usepackage[T1]{fontenc}    
\usepackage[space]{grffile}
\usepackage[utf8]{inputenc} 
\usepackage{algorithm}
\usepackage{algpseudocode}
\usepackage{amsfonts}       
\usepackage{amsmath}
\usepackage{amssymb}
\usepackage{amsthm}
\usepackage{arxiv}
\usepackage{booktabs}       
\usepackage{color}
\usepackage{graphicx}
\usepackage{hyperref}       
\usepackage{lipsum}		
\usepackage{microtype}      
\usepackage{nicefrac}       
\usepackage{url}            
\usepackage{verbatim}

\newtheorem{theorem}{Theorem}
\newtheorem{definition}{Definition}
\newtheorem{example}{Example}
\makeatletter

\providecommand{\tabularnewline}{\\}

\title{
  Analysis of Regression Tree Fitting Algorithms \\
  in Learning to Rank
}

\author{
  Tian Xia \\
  Wright State University \\
  \texttt{SummerRainET2008@gmail.com} \\
  \And
  Shaodan Zhai \\
  Wright State University \\
  \texttt{ShaodanZhai@gmail.com} \\
  \And
  Shaojun Wang \\
  Wright State University \\
  \texttt{SWang.USA@gmail.com} \\
}

\begin{document}
  \maketitle

  \begin{abstract}
    In learning to rank area, industry-level applications have been dominated by
    gradient boosting framework, which fits a tree using least square error 
    principle. While in classification area, another tree fitting principle,
    weighted least square error, has been widely used, such as
    LogitBoost and its variants. However, there is a lack of analysis on the
    relationship between the two principles in the scenario of learning to rank.
    We propose a new principle named least objective loss based error that
    enables us to analyze the issue above as well as several important learning
    to rank models. We also implement two typical and strong systems and
    conduct our experiments in two real-world datasets. Experimental results
    show that our proposed method brings moderate improvements over least square
    error principle.

  \end{abstract}

  \section{Introduction}

  Top practical learning to ranking systems are adopting gradient boosting
  framework and using regression trees as weak learners. These systems
  performed much better than linear systems on real-world datasets, such
  as Yahoo challenge 2010 (\cite{chapelle2011yahoo}), and Microsoft 30K
  (\cite{tan2013direct}). Among these systems, LambdaMART
  (\cite{wu2010adapting,burges2011learning}), a pair-wise
  based model, gained an excellent reputation in Yahoo challenge; MART
  (\cite{friedman2001greedy}), a point-wise based, is a regression model
  which utilizes least square loss as objective loss function, and McRank
  (\cite{li2007mcrank})\footnote{Li et al. call the model of McRank as MART in the
  scenario of classification.},
  a point-wise based, uses multi-class classification technique and
  converts predictions into ranking. For industry application, gradient
  boosting combined with regression trees appears to be a standard practice. 

  An important finding was made in the work of (\cite{cossock2006subset,li2007mcrank})
  that has created a bridge between learning to rank and classification. They
  proved that an important measure NDCG in learning to rank could be bounded
  by multi-class classification error. 
  This insight opens a door for learning to rank, as we could borrow
  state-of-the-art classification techniques.

  Least square error (SE) principle in fitting a regression tree only utilizes
  the first-order information of objective loss function, while in multi-class 
  classification area, there is a work that fits a regression tree by use of 
  second-order information besides that. Their tree fitting principle is called 
  weighted least square error (WSE). LogitBoost (\cite{friedman2000special}) and 
  its robust versions (\cite{li2010fast,li2012robust}) are examples.
  A comparison between gradient boosting using SE and
  LogitBoost using WSE for classification task (\cite{friedman2001greedy})
  shows that the latter is slightly better. As WSE is empirically considered
  as unstable in practice (\cite{friedman2000special,li2012robust,
  friedman2001greedy}), \cite{li2012robust} obtained a stable form of 
  WSE, called RWSE. 

  \begin{figure}[t]
    \begin{centering}
      \includegraphics[scale=0.3]{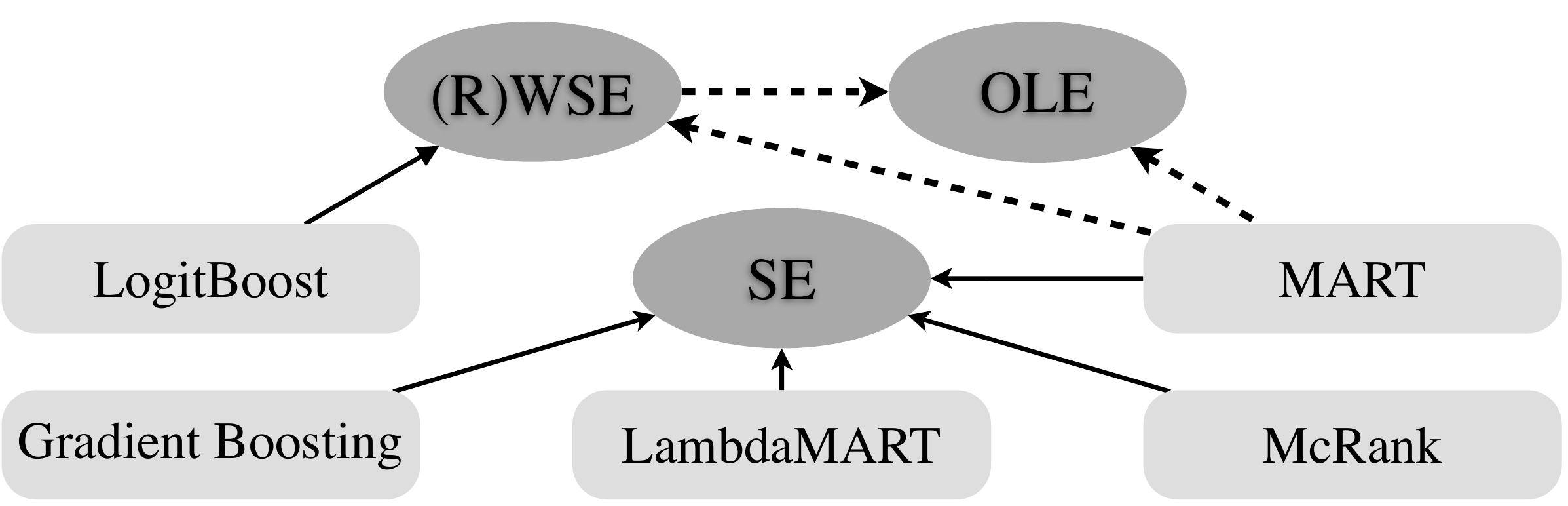}
      \par
    \end{centering}

    \protect\caption{
      Dark shadow marked ellipses denote tree fitting principles, and
      rectangles denote ranking models. SE: least square error; RWSE: robust
      weighted square error; OLE: objective loss based error in this
      paper. All real-line arrows denote known relations, and dotted-line
      arrows denote our contribution.
    }
    \label{fig_relation}
  \end{figure}

  However, both WSE and RWSE have no clear theoretical explanation, and somewhat
  hard to understand. As a result, Li thus proposed an interesting
  question in Section 2.3 of (\cite{li2012robust}): \emph{in determining
  the output of a leaf of a regression tree, the one-step Newton formula
  from gradient boosting coincides with the weighted averaging from
  WSE.} Moreover, RWSE looks concise, which might be mistaken to apply to any
  ranking model besides LogitBoost and its variants.
  These issues drive us to consider from another point of view.

  We propose a general regression tree fitting principle for ranking
  models, called least objective loss based error (OLE). It only requires
  simple computation to derive exact formula and is easy to understand.
  Under this principle, we first clearly answer the aforementioned
  question, and then analyze a variety of ranking systems to build a relationship
  between SE, (R)WSE and OLE (Figure \ref{fig_relation}).
  Experiments in real-world datasets also show moderate improvements.

  \section{Background}

  Given a set of queries $\mathbf{Q}=\{q_{1},\ldots,q_{|\mathbf{Q}|}\}$,
  each query $q_{i}$ is associated with a set of candidate relevant
  documents $\mathbf{D}_{i}=\{d_{1}^{i},\ldots,d_{|\mathbf{D}_{i}|}^{i}\}$
  and a corresponding vector of relevance scores $\mathbf{Y}_{i}=\{y_{i}(d_{1}^{i}),\ldots,y_{i}(d_{|\mathbf{D}_{i}|}^{i})\}$
  for each $\mathbf{D}_{i}$. The relevance score is usually an integer,
  assuming $0 \le y \le K$. Greater value means more related for the document
  to the query. An $M$-dimensional feature vector $\mathbf{h}(d)=[h_{1}(d|q),\ldots,h_{M}(d|q)]^{T}$
  is created for each query-document pair, where $h_{t}(\cdot)$s are
  predefined feature functions, which usually return real values.

  A ranking function $f$ is designed to score each query-document pair,
  and the documents associated with the same query are ranked by their
  scores and returned to users. Since these documents have a fixed ground
  truth rank with its corresponding query, our goal is to find an optimal
  ranking function which returns such a rank of related documents that
  is as close to the ground truth rank as possible. Industry-level applications
  often adopt regression trees to construct the ranking function, and
  use Newton Formula to calculate the output values of leaves of trees.

  Since fitting a regression tree is conducted node by node, in this
  paper we only discuss the differences of algorithms running on each
  node. We call them node-splitting algorithms. To distinguish them 
  from the objective loss functions of ranking models, we refer to
  the loss of each node-splitting algorithm as the "error".  For
  example, square error (Section \ref{sec_least_squared_splitting}),
  (robust) weighted square error (Section \ref{sec_weighted_least_squared},
  \ref{sec_robust_weighted_least}), and the objective loss based error
  (Section \ref{sec_Objective_Loss_Based}) proposed in this work. Obviously,
  if there is a minor difference in the split of a tree node for two
  node-splitting algorithms, they would lead to two totally different
  regression trees. 

  Several measures have been used to quantify the quality of the rank
  with respect to a ground-truth rank such as NDCG, ERR, MAP etc. In
  this paper, we use the most popular NDCG and ERR (\cite{chapelle2011yahoo})
  as the performance measures.

  \subsection{Square Error (SE) in Gradient Boosting \label{sec_least_squared_splitting}}

  We regard gradient boosting (\cite{friedman2001greedy}) as a general
  framework for function approximation, and it is applicable for any differentiable
  loss function. Combining with regression
  trees as weak learners has been the most successful application in
  learning to rank area.

  Gradient boosting iteratively finds an additive predictor $f(\cdot)\in\mathcal{H}$,
  which minimizes the loss function $\mathcal{L}$. At the $t$th iteration,
  a new weak learner $g_{t}(\cdot)$ is selected to add to current predictor
  $f_{t}(\cdot)$. Then

  \begin{equation}
    f_{t+1}(\cdot)=f_{t}(\cdot)+\alpha g_{t}(\cdot)\label{eq_update}
  \end{equation}
  where $\alpha$ is the learning rate. 

  In gradient boosting, $g_{t}(\cdot)$ is chosen as the one most parallel
  to pseudo-response $r(\cdot)$, which is defined as the negative derivative
  of the loss function in functional space $\mathcal{H}$.

  \begin{equation}
    r(\cdot)=-\mathcal{L}'(f_{t}(\cdot))\label{eq_g_def_in_sqrt}
  \end{equation}

  In practice, the globally optimal $g_{t}(\cdot)$ usually can not be obtained, then 
  we turn to operations on each tree node to fit a sub-optimal tree. 
  In one node, we enumerate all possible feature-threshold pairs, and then 
  the best one is selected to conduct binary splitting. This procedure recursively iterates on its child
  nodes until reaching a predefined condition.

  Regarding a feature, suppose a threshold $v$ to split samples on
  the current node into two parts. The samples, whose feature values
  are less than $v$, are denoted as $\mathbf{D}_{l}$, and others are
  denoted as $\mathbf{D}_{r}$. The squared error is defined as

  \begin{equation}
    SE(v)=\sum_{d\in\mathbf{D}_{l}}(r(d)-\overline{r_{1}})^{2}+\sum_{d\in\mathbf{D}_{r}}(r(d)-\overline{r_{2}})^{2}\label{eq_sqrt_loss}
  \end{equation}
  where $\overline{r_{1}}$, $\overline{r_{2}}$ are average pseudo-response
  of samples on the left and right respectively.

  The complexity of a regression tree could be controlled by limiting the tree height
  or leaf number. In learning to rank, the latter is more flexible,
  and adopted in this work by default. 

  \subsection{Weighted Square Error in Classification}

  \cite{cossock2006subset} proved that the negative unnormalized NDCG value is upper-bounded by multi-class
  classification error, where NDCG is an important measure in learning
  to rank. Thus \cite{li2007mcrank} proposed a multi-class
  classification based ranking systems $ $called McRank in gradient
  boosting framework.

  McRank utilizes classic logistic regression, which models class probability
  $p_{k}(d)$ as

  \begin{equation}
    p_{k}(d)=p(y(d)=k)=\frac{\exp f^{k}(d)}{\sum_{c=0}^{K}\exp f^{c}(d)}
  \end{equation}
  where $f^{c}(\cdot)$ is an additive predictor function for the $c$th
  class.

  The objective loss function is the negative log-likelihood, defined as

  \begin{equation}
    \mathcal{L}=-{\sum_{d}}{\sum_{k=0}^{K}\,}\textrm{I}(y(d)=c)\log p_{k}(d)
    \label{eq_log_loss}
  \end{equation}
  where $\textrm{I}(\cdot)$ is an indicator function.

  \subsubsection{Weighted Square Error (WSE)\label{sec_weighted_least_squared}}

  In classification, this loss function (Eqn. \ref{eq_log_loss}) resulted
  in the well-known system LogitBoost, which first used WSE to fit a
  regression tree. WSE utilizes both first- and second- order derivative
  information.

  WSE uses a different definition of the response value $r(\cdot)$
  from that in SE (Eqn. \ref{eq_g_def_in_sqrt}), and defines an extra
  weight $w(\cdot)$ for each sample. 

  \begin{equation}
    \begin{array}{clc}
      r(\cdot) & = & -\mathcal{L}'(f_{t}(\cdot))/\mathcal{L}''(f_{t}(\cdot))\\
      w(\cdot) & = & \mathcal{L}''(f_{t}(\cdot))
    \end{array}\label{eq_r_w_def_in_wsqrt}
  \end{equation}

  The splitting principle is minimizing the following weighted error

  \begin{equation}
    WSE(v) 
    = \left[\sum_{d\in\mathbf{D}_{l}}w(d)(r(d)-\overline{r_{1}})^{2}+\sum_{d\in\mathbf{D}_{r}}w(d)(r(d)-\overline{r_{2}})^{2}\right]
    - \sum_{d\in\mathbf{D}}w(d)(r(d)-\bar{r})^{2}\label{eq_wsqrt_loss}
  \end{equation}
  where

  \begin{equation}
    \overline{r_{1}} =\frac{\sum_{d\in\mathbf{D}_{l}}w(d)\cdot r(d)}{\sum_{d\in\mathbf{D}_{l}}w(d)} 
    ~~ \overline{r_{2}}=\frac{\sum_{d\in\mathbf{D}_{r}}w(d)\cdot r(d)}{\sum_{d\in\mathbf{D}_{r}}w(d)} 
    ~~ \overline{r} =\frac{\sum_{d\in\mathbf{D}}w(d)\cdot r(d)}{\sum_{d\in\mathbf{D}}w(d)}
    \label{eq_wsqrt_loss_w}
  \end{equation}

  However, Logitboost is usually thought as instable. Combining its loss function, the response and weight
  values, by Eqn. \ref{eq_r_w_def_in_wsqrt},
  are set as $\ensuremath{w(d)=p_{k}(d)(1-p_{k}(d))}$, $\ensuremath{r(d)=\frac{\textrm{I}(y(d)=k)-p_{k}(d)}{p_{k}(d)(1-p_{k}(d))}}$,
  in fitting a tree for the $k$th classification. The response $r(d)$
  might become huge and lead to unsteadiness when $p_{k}(d)$ in the denominator
  is close to $0$ or $1$. Though \cite{friedman2000special} described
  some heuristics to smooth the response values, LogitBoost was still
  believed numerically unstable (\cite{friedman2000special,friedman2001greedy,li2012robust}).
  As a result, McRank actually adopts SE and gradient boosting to
  fit regression trees (Section \ref{sec_least_squared_splitting}),
  rather than LogitBoost.

  \subsubsection{Robust Weighted Square Error (RWSE)\label{sec_robust_weighted_least}}

  \cite{li2012robust} derived Eqn. \ref{eq_wsqrt_loss} and proposed
  a stable version of WSE for LogitBoost, which is shown below

  \begin{equation}
    RWSE(v)
    = -\left[\frac{[\sum_{d\in\mathbf{D}_{l}}w(d)r(d)]^{2}}{\sum_{d\in\mathbf{D}_{l}}w(d)}+\frac{[\sum_{d\in\mathbf{D}_{r}}w(d)r(d)]^{2}}{\sum_{d\in\mathbf{D}_{r}}w(d)}\right]
    + \frac{[\sum_{d\in\mathbf{D}}w(d)r(d)]^{2}}{\sum_{d\in\mathbf{D}}w(d)}\label{eq_rwsqrt_loss}
  \end{equation}

  Since all denominators in Eqn. \ref{eq_rwsqrt_loss} are summation
  of a set of weights $w(\cdot)$, which are less likely to be close
  to zero in practical applications, and RWSE is hence more stable than
  WSE.

  After fitting a regression tree by either SE or (R)WSE, the data aggregated
  in the same leaf is assigned with a value by weighted averaging responses
  (Eqn. \ref{eq_wsqrt_loss_w}).

  Li proposed an interesting question regarding Eqn. \ref{eq_wsqrt_loss_w}.
  Li mentioned, Eqn. \ref{eq_wsqrt_loss_w} could be interpreted as
  a weighted average in (R)WSE; while in gradient boosting, it is interpreted
  as a one-step Newton update. It looks like a coincidence. In next
  section, we propose a unified splitting principle, which not only
  clearly explains the relationship of these principles, but also could
  be extended to more complex loss functions. Also, our method generates
  Li's robust version directly.

  \section{Greedy Tree Fitting Algorithm in Learning to Rank}

  \subsection{Objective Loss Based Error (OLE)\label{sec_Objective_Loss_Based}}

  We were motivated by the success of AdaBoost (\cite{schapire2012boosting}). In each iteration,
  AdaBoost selects the weak learner that has a minimal weighted error, and this
  can be proved that \emph{the weak learner selected
  ensures a maximum improvement of its objective loss}, so we are borrowing a similar strategy 
  in learning to rank. Note that, we would use totally different formulas.

  Exactly finding an optimal regression tree is computationally infeasible,
  as the number of possible trees is combinatorially huge. We thus turn
  to focus on the most basic unit in fitting a regression tree, that
  is how to conduct a good binary partition to improve the objective
  loss most. This is a more acceptable approach.

  Given a set of samples $\mathbf{D}=\{d_{1},\ldots,d_{|\mathbf{D}|}\}$,
  and a selected feature,  we first assume there are at most $|\mathbf{D}|-1$
  potential positions to define a threshold $v$ for the selected feature.
  Based on a threshold $v$, the samples are split into two parts, $\mathbf{D}_{l}$
  and $\mathbf{D}_{r}$. Second, we assume once a partition is conducted,
  the samples on the two sides would receive their updated score. In
  other words, we fit a temporary two-leaf tree in the current samples
  and update the outputs of two leaves separately (diagonal approximation).
  We update samples and calculate the objective loss, and allow the
  maximum improvement quantity as a measure for the current partition.
  The best partition and its respective threshold are selected after
  enumerating at most $|\mathbf{D}|-1$ possibilities. We ignore the
  fact that either side may not be a real leaf in the final fitted regression
  tree, so that we have a feasible method.

  Regarding a threshold $v$, let the outputs of the temporary two-leaf
  tree be $o_{1}$ and $o_{2}$, then the objective loss $\mathcal{L}$
  has become a function of $o_{1}$ and $o_{2}$. Once the values of
  $o_{1}$ and $o_{2}$ are determined, samples on two sides would be
  updated, and then the objective loss can be straightforwardly computed.
  However, even in a moderate size dataset, this computation is still
  prohibitive. So we approximate the objective loss with the Taylor
  formula in the second order at the point of 0.

  \begin{equation}
    \mathcal{L}(o) = \mathcal{L}(o = 0) + 
    o \cdot \mathcal{L}'(o = 0) + 
    \frac{o^{2}}{2}\mathcal{L}''(o = 0)\label{eq_taylor}
  \end{equation}
  where $o\in\{o_{1},o_{2}\}$.

  The local optimum $o$ can be obtained as $-\frac{\mathcal{L}'(o=0)}{\mathcal{L}''(o=0)}$
  by letting the first-order derivative $\mathcal{L}'(o)$ be zero.
  More specifically, $o_{1}$ and $o_{2}$ are optimized independently as following

  \begin{equation}
    \begin{cases}
      o_{1}=-\frac{\mathcal{L}'(o_{1}=0)}{\mathcal{L}''(o_{1}=0)} & \textrm{if}\,d\in\textrm{D}_{l}\\
      o_{2}=-\frac{\mathcal{L}'(o_{2}=0)}{\mathcal{L}''(o_{2}=0)} & \textrm{if}\,d\in\textrm{D}_{r}
    \end{cases}\label{eq_g_d}
  \end{equation}

  Insert Eqn. \ref{eq_g_d} into Eqn. \ref{eq_taylor}, and simplify
  to obtain our objective loss based error

  \begin{align}
    OLE(v)= & \mathcal{L}(o_{1}=0,o_{2}=0)
    -\frac{1}{2}\left[\frac{\mathcal{L}'(o_{1}=0)^{2}}{\mathcal{L}''(o_{1}=0)}+\frac{\mathcal{L}'(o_{2}=0)^{2}}{\mathcal{L}''(o_{2}=0)}\right]\\
    \propto & -\left[\frac{\mathcal{L}'(o_{1}=0)^{2}}{\mathcal{L}''(o_{1}=0)}+\frac{\mathcal{L}'(o_{2}=0)^{2}}{\mathcal{L}''(o_{2}=0)}\right]\label{eq_OL}
  \end{align}

  This resultant formula is not equivalent to SE or (R)WSE in a general
  case, and would lead to a totally different regression tree from that
  using other node-splitting algorithms.

  In the case of learning to rank, we analyze their equivalence for
  point-, pair- and list- wised based models.

  \subsection{Derivative Additive Loss Functions}

  In order to calculate all the gradients in Eqn. \ref{eq_OL} in an
  efficient left-to-right incremental updating way, we explore the cases
  the gradient $\mathcal{L}'(o=0)$, $\mathcal{L}''(o=0)$ can be decomposed
  into operations on each sample.
  \begin{definition}
    A loss function $\mathcal{L}$ is defined as derivative additive if
    $\mathcal{L}'(o=0|\mathbf{D})=\sum_{d\in\mathbf{D}}\mathcal{L}'(o=0|d)$
    and $\mathcal{L}''(o=0|\mathbf{D})=\sum_{d\in\mathbf{D}}\mathcal{L}''(o=0|d)$.
  \end{definition}
  \begin{example}
    The loss function of MART system is derivative additive, since 
  \end{example}

  \[
  \begin{array}{ll}
    \mathcal{L} & =\sum_{d}(f_{t}(d)-y(d))^{2}\\
    \\
    \mathcal{L}'(o=0) & =\frac{\partial\sum_{d}(f_{t}(d)+o-y(d))^{2}}{\partial o}=\sum_{d}2(f_{t}(d)-y(d))\\
    & =\sum_{d}\mathcal{L}'(o=0|d)\\
    \mathcal{L}''(o=0) & =\frac{\partial\sum_{d}(f_{t}(d+o-y(d))^{2}}{\partial^{2}o}=\sum_{d}2\\
    & =\sum_{d}\mathcal{L}''(o=0|d)
  \end{array}
  \]

  \begin{example}
    The loss function of McRank system is derivative additive, since
  \end{example}

  \[
  \begin{array}{ll}
    \mathcal{L} & ={\displaystyle \sum_{d}}{\displaystyle \sum_{c=0}^{K-1}}\textrm{I}(c=y(d))\log p_{c}(d)\\
    \\
    \mathcal{L}'(o=0) & =\sum_{d}\sum_{c=0}^{K-1}\{p_{c}(d)-I(c=y(d))\}\\
    & =\sum_{d}\mathcal{L}'(o=0|d)\\
    \mathcal{L}''(o=0) & =\sum_{d}\sum_{c=0}^{K-1}p_{c}(d)(1-p_{c}(d))\\
    & =\sum_{d}\mathcal{L}''(o=0|d)
  \end{array}
  \]

  \begin{example}
    The loss function of RankBoost system is not derivative additive,
    since
  \end{example}

  \[
  \mathcal{L}=\sum_{\mathbf{D}_{i}}\sum_{d_{1},d_{2}\in\mathbf{D}_{i},y(d_{1})>y(d_{2})}\exp\{f_{t}(d_{2})-f_{t}(d_{1})\}
  \]

  To clearly explain, we use a toy example with $d_{1},d_{2},d_{3}\in\mathbf{D}_{1}$,
  sorted by their relevances $y(d_{1})>y(d_{2})>y(d_{3})$. Assuming
  in current $t$th iteration, their score is $f_{t}(\cdot)$. The exponential
  loss is

  \[
  \begin{array}{ll}
    \mathcal{L} &
    =\exp\{f_{t}(d_{2})-f_{t}(d_{1})\}+\exp\{f_{t}(d_{3})-f_{t}(d_{1})\} + \exp\{f_{t}(d_{3})-f_{t}(d_{2})\}\} \\
    & =s_{1}+s_{2}+s_{3}
  \end{array}
  \]
  where, to simplify, $s_{1}=\exp\{f_{t}(d_{2})-f_{t}(d_{1})\}$, $s_{2}=\exp\{f_{t}(d_{3})-f_{t}(d_{1})\}$,
  $s_{3}=\exp\{f_{t}(d_{3})-f_{t}(d_{2})\}$.

  \[
  \begin{array}{ll}
    \textrm{as }\mathcal{L}(o|d_{1})=s_{1}\exp\{-o\} & +s_{2}\exp\{-o\}+s_{3}\\
    \textrm{so }\mathcal{L}'(o=0|d_{1})=-s_{1}-s_{2} & \mathcal{L}''(o=0|d_{1})=s_{1}+s_{2}\\
    \textrm{likewise }\mathcal{L}'(o=0|d_{2})=s_{1}-s_{3} & \mathcal{L}''(o=0|d_{2})=s_{1}+s_{3}\\
    \textrm{likewise }\mathcal{L}'(o=0|d_{3})=s_{2}+s_{3} & \mathcal{L}''(o=0|d_{3})=s_{2}+s_{3}
  \end{array}
  \]

  Suppose one partition is $\{d_{1},d_{2}\}$ and $\{d_{3}\}$, with
  the output value $o_{1}$, $o_{2}$ respectively, then the current
  loss is

  \[
  \begin{array}{ll}
    \mathcal{L}(o_{1},o_{2}) & =\exp\{(f_{t}(d_{2})+o_{1})-(f_{t}(d_{1})+o_{1})\}\\
    & +\exp\{(f_{t}(d_{3})+o_{2})-(f_{t}(d_{1})+o_{1})\}\\
    & +\exp\{(f_{t}(d_{3})+o_{2})-(f_{t}(d_{2})+o_{1})\}\}\\
    & +const\\
    \mathcal{L}'(o_{1}=0) & =-s_{2}-s_{3}\\
    \mathcal{L}''(o_{1}=0) & =s_{2}+s_{3}\\
    & \ne\mathcal{L}''(o_{1}=0|d_{1})+\mathcal{L}''(o_{2}=0|d_{2})\\
    & =s_{1}+s_{1}+s_{2}+s_{3}
  \end{array}
  \]

  The key reason is that \emph{if two samples appearing in the same $\exp$
  term of the objective loss also are classified into the same leaf,
  then they would receive the same output from current leaf, which does
  not contribute to the objective loss}. In this example, the two $s_{1}$,
  coming from $d_{1}$ and $d_{2}$, should be counteracted. As a result,
  the exponential loss is not derivative additive.
  \begin{example}
    The loss function of LambdaMART system is not derivative additive. 
  \end{example}
  This famous system has no explicit objective loss function, but has
  exact first- and second-order derivatives. Its first derivative has
  a similar unit with the RankBoost model, both having such terms of
  $\exp(f_{t}(d_{i})-f_{t}(d_{j}))$, since LambdaMART is pair-wise
  based system. Based on the detailed analysis in RankBoost, we could
  easily know the loss function of LambdaMART, potentially existing,
  is not derivative additive\footnote{To examine strictly, the one-step Newton formula in LambdaMART (Line
  11 in Alg. 1 of (\cite{wu2010adapting})) is incorrect conceptually,
  as the denominator in $\frac{\mathcal{L}'(o=0)^{2}}{\mathcal{L}''(o=0)}$
  is tackled as derivative additive. We have not found any explanations
  from their paper. But it can be viewed as an approximation of the
  exact formula.}. 
  \begin{example}
    The loss function of ListMLE is not derivative additive. 
  \end{example}
  As the term $\exp(f_{t}(d_{i})-f_{t}(d_{j}))$ is frequently appearing
  in its loss function, the loss function of ListMLE is also not derivative
  additive. 
  \begin{example}
    Some special list-wise models have derivative additive loss functions.
  \end{example}
  In the work of (\cite{ravikumar2011ndcg}), several point-wise based systems are
  modified by using list-wise information, so they are considered to
  be list-wise based systems, such as consistent-MART, consistent KL
  divergence based, consistent cosine distance based. As the extra list-wise
  information is actually utilized in a preprocessing step, and then
  they are running in a point-wise style, so these so-called list-wise
  based systems also own derivative additive loss functions. 

  \subsection{(R)WSE $\subset$ OLE}

  (R)WSE was proposed for LogitBoost, which is a classification system.
  However, from the angle of learning to rank, LogitBoost could be classified into
  point-wise based. We prove (R)WSE is actually a special case of OLE.
  \begin{theorem}
    Regarding derivative additive loss function, OLE is simplified into
    (R)WSE.\label{thm_Regarding_derivative_additive}
  \end{theorem}

  \begin{proof}
    \[
    \begin{array}{ll}
      OLE(v) & (\textrm{Eqn. \ref{eq_OL}})\\
      = & -\left[\frac{\left[\sum_{d\in\mathbf{D}_{l}}\mathcal{L}'(o_{1}=0|d)\right]^{2}}{\sum_{d\in\mathbf{D}_{l}}\mathcal{L}''(o_{1}=0|d)}+\frac{\left[\sum_{d\in\mathbf{D}_{r}}\mathcal{L}'(o_{2}=0|d)\right]^{2}}{\sum_{d\in\mathbf{D}_{r}}\mathcal{L}''(o_{2}=0|d)}\right]\\
      & \textrm{by Eqn.\ref{eq_r_w_def_in_wsqrt}}\\
      \propto & -\left[\frac{\left[\sum_{d\in\mathbf{D}_{l}}w(d)\cdot r(d)\right]^{2}}{\sum_{d\in\mathbf{D}_{l}}w(d)}+\frac{\left[\sum_{d\in\mathbf{D}_{r}}w(d)\cdot r(d)\right]^{2}}{\sum_{d\in\mathbf{D}_{r}}w(d)}\right]+\textrm{const}\\
      & \textrm{where const}=\frac{\left[\sum_{d\in\mathbf{D}}w(d)\cdot r(d)\right]^{2}}{\sum_{d\in\mathbf{D}}w(d)}\\
      = & RWSE(v)\textrm{ (Eqn. \ref{eq_rwsqrt_loss})}
    \end{array}
    \]

    We simply obtain the robust weighted least square error from (\cite{li2012robust}).
    As Li proved the robust version is equivalent to original WSE, thus
    our method is equivalent to WSE for all derivative additive loss
    functions.
  \end{proof}

  By the explanation of robustness of Li, our OLE method (Eqn. \ref{eq_OL})
  is intrinsically robust, as all denominators are less likely to be
  zero in summing a set of samples.

  Recall that a model is classified into the point-wise category if
  the model does not use the relationship between samples, but only
  individual samples. See the typical point-wise based models, Example
  1 and 2. So, point-wise based systems have derivative additive loss
  functions, and in this case, (R)WSE is always equivalent to OLE.

  A pair-wise based system considers the relationship only between two
  samples. If a pair of associated samples are classified into two different
  tree nodes, the objective loss function is derivative additive; otherwise,
  it is not. In practical applications, it is not difficult to overcome
  this inconvenience by using an incremental updating.

  A list-wise based system would render more samples interact to each
  other, and it is relatively more difficult to tackle. But in splitting
  a tree node, the incremental updating is still working.

  We are now able to answer the question from (\cite{li2012robust}),
  the Eqn. \ref{eq_wsqrt_loss_w} appears to have two explanations,
  one from weighted average, and the other from one-step Newton. As
  we proved that (R)WSE is a special case derived from optimizing only
  derivative additive objectives, and LogitBoost uses derivative additive
  objective, so (R)WSE and SE are equivalent. Moreover, for other complex
  objective losses, (R)WSE may have no theoretical support, but it may
  serve as an approximation of our method.

  \subsection{SE = (R)WSE = OLE for MART \label{sec_SE_RWSE_OLE_MARTE}}

  SE is generally not equivalent to (R)WSE or OLE, even the objective
  loss functions used satisfy the condition of Theorem \ref{thm_Regarding_derivative_additive}.
  However, we find that MART is an ideal intersection of OLE, (R)WSE
  and SE. 

  MART system adopts least square loss as objective loss, and classic
  gradient boosting framework to fit regression trees. Many commercial
  search engines are using this model to construct their ranking systems.
  \begin{theorem}
    Regarding MART system, whose objective loss is the least-square $\sum_{d}|f(d)-y(d)|^{2}$,
    then $\arg\underset{v}{\min}SE(v)=\arg\underset{v}{\min}(R)WSE(v)=\arg\underset{v}{\min}OLE(v)$.
    \label{thm_Regarding_MART_system}
  \end{theorem}

  \begin{proof}
    Given some chosen feature function $f(\cdot)$ and pseudo-response
    $r(\cdot)$ of each document $d\in\mathbf{D}$, there are $|\mathbf{D}|-1$
    positions to define a threshold which is the middle value of two adjacent
    feature values. 

    We derive from the definition to prove that minimizing objective loss
    is the same with minimizing the least square error splitting principle (Eqn.
    \ref{eq_sqrt_loss}).

    \begin{multline*}
      \begin{array}{rl}
        & \mathcal{L}(v)\\
        = & \sum_{d\in\mathbf{D}_{l}}(f_{t+1}(d)-y(d)){}^{2}+\sum_{d\in\mathbf{D}_{r}}(f_{t+1}(d)-y(d)){}^{2}\\
        \\
        = & \sum_{d\in\mathbf{D}_{l}}((f_{t}(d_{i})+o_{1})-y(d))^{2} + \sum_{d\in\mathbf{D}_{r}}((f_{t}(d)+o_{2})-y(d))^{2}\\
        \\
        & \textrm{ By the Newton formula}\\
        & o_{1}=-\frac{\mathcal{L}'(o_{1}=0)}{\mathcal{L}''(o1=0)}=-\frac{\sum_{d\in\mathbf{D}_{l}}\mathcal{L}'(f_{t}(d))}{2\cdot|\mathbf{D}_{l}|} 
        =\frac{\sum_{d\in\mathbf{D}_{l}}r(d)}{2\cdot|\mathbf{D}_{l}|}=\frac{\bar{r_{1}}}{2}\\
        & o_{2}=-\frac{\mathcal{L}'(o_{2}=0)}{\mathcal{L}''(o_{2}=0)}=-\frac{\sum_{d\in\mathbf{D}_{r}}\mathcal{L}'(f_{t}(d))}{2\cdot|\mathbf{D}_{r}|}
        =\frac{\sum_{d\in\mathbf{D}_{r}}r(d)}{2\cdot|\mathbf{D}_{r}|}=\frac{\bar{r_{2}}}{2}\\
        \\
        & \mathcal{L}(v)\\
        = & \sum_{d\in\mathbf{D}_{l}}(f_{t}(d)+\frac{\bar{r_{1}}}{2}-y(d))^{2}
        +\sum_{d\in\mathbf{D}_{r}}(f_{t}(d)+\frac{\bar{r_{2}}}{2}-y(d))^{2}\\
        \propto & \sum_{d\in\mathbf{D}_{l}}((2f_{t}(d)-2y(d))+\bar{r_{1}})^{2}
        +\sum_{d\in\mathbf{D}_{r}}((2f_{t}(d)-2y(d))+\bar{r_{2}})^{2}\\
        \\
        & \textrm{As }2f_{t}(d)-2y(d)=\mathcal{L}'(f_{t}(d))=-r(d)\\
        & \mathcal{L}(v)\\
        = & \sum_{d\in\mathbf{D}_{l}}(r(d)-\bar{r_{1}})^{2}+\sum_{d\in\mathbf{D}_{r}}(r(d)-\bar{r_{2}})^{2}\\
        = & \textrm{Eqn. \ref{eq_sqrt_loss}}
      \end{array}
    \end{multline*}

    So, optimizing objective loss is equivalent to optimizing SE, and
    as mentioned before MART is a point-wise bases system, which suggests
    $(R)WSE=OLE$.
  \end{proof}

  This theorem means, regarding MART system, the three node-splitting
  principles lead to the same binary partition in any selected tree
  node, and then lead to the same regression tree in current iteration
  of boosting. 

  The key step here is \emph{the average pseudo-response for $\bar{r_{1}}$
  and $\bar{r_{2}}$, whose definition is same with that in Eqn. \ref{eq_sqrt_loss},
  is exactly double of the negative optimum computed by Newton equation
  using the extra second derivative}. Regarding other loss functions,
  this relationship does not necessarily hold. 

  This theorem shows the classic tree fitting algorithm in gradient
  boosting is very suitable for the least-square loss function, on which
  MART system is based, and this could explain why MART system actually
  performs excellently in many practical applications. 

  Besides, there is a by product formula for MART system. 
  By plugging its first-order derivative $\mathcal{L}'(o_{1}=0)=\sum_{d\in\mathbf{D_{l}}}-r(d)$,
  $\mathcal{L}'(o_{2}=0)=\sum_{d\in\mathbf{D_{r}}}-r(d)$, and second-order
  derivative $ $$\mathcal{L}''(o_{1})=2|\mathbf{D}_{l}|$ , $\mathcal{L}''(o_{2})=2|\mathbf{D}_{l}|$
  into Eqn. \ref{eq_OL}, we obtain a simpler splitting principle than Eqn.
  \ref{eq_sqrt_loss}.

  \begin{equation}
    MART(v)=-\left[\sum_{d\in\mathbf{D}_{l}}\frac{r(d)^{2}}{|\mathbf{D}_{l}|}+\sum_{d\in\mathbf{D}_{r}}\frac{r(d)^{2}}{|\mathbf{D}_{r}|}\right]
  \end{equation}

  This form is more intuitive for incremental computation of the optimal
  threshold $v$ from left to right.

  \section{Experiments}

  \subsection{Datasets and Systems}

  As suggested by \cite{qin2010letor}, we use two real-world datasets
  to make our results more stable, Yahoo challenge 2010 and Microsoft
  10K. The statistics of these data sets are reported in Table \ref{tb_dataset}. 
  \begin{enumerate}
    \item \emph{Yahoo Challenge 2010}. After Yahoo corporation hosted this
    far-reaching influence contest of learning to rank in 2010, this dataset
    has been important for a comparison. It contains two sets, and here
    we use the bigger one (set 1). Yahoo dataset was released with only
    one split of training, validating, and testing set, and \emph{we
    add an extra two splits and also report average results}. 

    \item \emph{Microsoft 10K}. Another publicly released datasets, and even
    larger than the Yahoo data in terms of the number of documents. As
    a 5-fold splitting is provided by official release, we report average
    results. 
  \end{enumerate}

  \begin{table}[t]
    \begin{centering}
      \begin{tabular}{c|cccc}
        & \textbf{\footnotesize \#Query} & \textbf{\footnotesize \#Doc.} & \textbf{\footnotesize \#D. / \#Q.} & \textbf{\footnotesize \#Feat.}\tabularnewline
        \hline 
        \textbf{\footnotesize Yahoo} & {\footnotesize 20K} & {\footnotesize 473K} & {\footnotesize 23} & {\footnotesize 519}\tabularnewline
        \textbf{\footnotesize Micro-10K} & {\footnotesize 6K} & {\footnotesize 723K} & {\footnotesize 120} & {\footnotesize 136}\tabularnewline
        \hline 
        \textbf{\footnotesize McRank ((\cite{li2007mcrank}))} & {\footnotesize 10-26K} & {\footnotesize 474-1741K} & {\footnotesize 18-88} & {\footnotesize 367-619}\tabularnewline
        \textbf{\footnotesize $\lambda$-MART ((\cite{wu2010adapting}))} & {\footnotesize 31K} & {\footnotesize 4154K} & {\footnotesize 134} & {\footnotesize 416}\tabularnewline
        \textbf{\footnotesize Ohsumed} & {\footnotesize 106} & {\footnotesize 16K} & {\footnotesize 150} & {\footnotesize 45}\tabularnewline
        \textbf{\footnotesize letor 4.0} & {\footnotesize 2.4K} & {\footnotesize 85K} & {\footnotesize 34} & {\footnotesize 46}\tabularnewline
      \end{tabular}
      \par
    \end{centering}

    \protect\caption{
      The top two datasets are used in this work and others are just as
      reference. $\lambda$-MART is LambdaMART. Ohsumed belongs to letor 3.0.
      \#D./\#Q. means average document number per query.
    }
    \label{tb_dataset}
  \end{table}

  The two datasets above were empirically found to be different. The
  Microsoft dataset seems more difficult than Yahoo as some models are
  reportedly running badly on it (\cite{tan2013direct}). It has comparatively
  less features, 136, and larger average number of documents per query
  120, compared to 519 and 23 of Yahoo. The two real-world datasets
  should be capable of providing convincing results.

  \begin{table*}
    \begin{centering}
      \begin{tabular}{c|c|c|cccc}
        \multicolumn{1}{c}{} &  &  & \multicolumn{4}{c}{\textbf{\footnotesize{}Yahoo challenge 2010}}\tabularnewline
        \hline 
        & \textbf{\footnotesize{}\#leaf} & \textbf{\footnotesize{}$\alpha$} & \textbf{\footnotesize{}NDCG@1} & \textbf{\footnotesize{}NDCG@3} & \textbf{\footnotesize{}NDCG@10} & \textbf{\footnotesize{}ERR}\tabularnewline
        \hline 
        &  & {\footnotesize{}0.06} & {\footnotesize{}71.32/}\textbf{\footnotesize{}71.76} & {\footnotesize{}71.47/}\textbf{\footnotesize{}72.22} & {\footnotesize{}77.97/}\textbf{\footnotesize{}78.52} & {\footnotesize{}45.44/}\textbf{\footnotesize{}45.71}\tabularnewline
        & {\footnotesize{}10} & {\footnotesize{}0.10} & {\footnotesize{}71.38/}\textbf{\footnotesize{}71.73} & {\footnotesize{}71.82/}\textbf{\footnotesize{}72.37} & {\footnotesize{}78.25/}\textbf{\footnotesize{}78.67} & {\footnotesize{}45.52/}\textbf{\footnotesize{}45.77}\tabularnewline
        \textbf{\footnotesize{}McRank} &  & {\footnotesize{}0.12} & {\footnotesize{}71.67/}\textbf{\footnotesize{}71.87} & {\footnotesize{}71.96/}\textbf{\footnotesize{}72.52} & {\footnotesize{}78.33/}\textbf{\footnotesize{}78.77} & {\footnotesize{}45.59/}\textbf{\footnotesize{}45.79}\tabularnewline
        \cline{2-7} 
        &  & {\footnotesize{}0.06} & {\footnotesize{}71.52/}\textbf{\footnotesize{}71.90} & {\footnotesize{}71.80/}\textbf{\footnotesize{}72.60} & {\footnotesize{}78.26/}\textbf{\footnotesize{}78.84} & {\footnotesize{}45.54/}\textbf{\footnotesize{}45.83}\tabularnewline
        & {\footnotesize{}20} & {\footnotesize{}0.10} & {\footnotesize{}71.65/}\textbf{\footnotesize{}72.03} & {\footnotesize{}72.14/}\textbf{\footnotesize{}72.77} & {\footnotesize{}78.50/}\textbf{\footnotesize{}79.00} & {\footnotesize{}45.64/}\textbf{\footnotesize{}45.88}\tabularnewline
        &  & {\footnotesize{}0.12} & {\footnotesize{}71.80/}\textbf{\footnotesize{}72.10} & {\footnotesize{}72.23/}\textbf{\footnotesize{}72.79} & {\footnotesize{}78.58/}\textbf{\footnotesize{}78.99} & {\footnotesize{}45.66/}\textbf{\footnotesize{}45.89}\tabularnewline
        \hline 
        &  & {\footnotesize{}0.06} & {\footnotesize{}71.15/}\textbf{\footnotesize{}71.62} & {\footnotesize{}71.60/}\textbf{\footnotesize{}71.94} & {\footnotesize{}77.82/}\textbf{\footnotesize{}78.15} & {\footnotesize{}45.66/}\textbf{\footnotesize{}45.80}\tabularnewline
        & {\footnotesize{}10} & {\footnotesize{}0.10} & {\footnotesize{}71.29/}\textbf{\footnotesize{}71.81} & {\footnotesize{}71.76/}\textbf{\footnotesize{}72.11} & {\footnotesize{}77.96/}\textbf{\footnotesize{}78.29} & {\footnotesize{}45.72/}\textbf{\footnotesize{}45.90}\tabularnewline
        \textbf{\footnotesize{}LambdaMART} &  & {\footnotesize{}0.12} & {\footnotesize{}71.30/}\textbf{\footnotesize{}71.76} & {\footnotesize{}71.76/}\textbf{\footnotesize{}72.19} & {\footnotesize{}77.93/}\textbf{\footnotesize{}78.34} & {\footnotesize{}45.67/}\textbf{\footnotesize{}45.87}\tabularnewline
        \cline{2-7} 
        &  & {\footnotesize{}0.06} & {\footnotesize{}71.51/}\textbf{\footnotesize{}71.75} & {\footnotesize{}72.02/}\textbf{\footnotesize{}72.25} & {\footnotesize{}78.13/}\textbf{\footnotesize{}78.40} & {\footnotesize{}45.77/}\textbf{\footnotesize{}45.92}\tabularnewline
        & {\footnotesize{}20} & {\footnotesize{}0.10} & {\footnotesize{}71.37/}\textbf{\footnotesize{}72.10} & {\footnotesize{}71.92/}\textbf{\footnotesize{}72.56} & {\footnotesize{}78.04/}\textbf{\footnotesize{}78.58} & {\footnotesize{}45.72/}\textbf{\footnotesize{}46.02}\tabularnewline
        &  & {\footnotesize{}0.12} & {\footnotesize{}71.44/}\textbf{\footnotesize{}71.76} & {\footnotesize{}71.91/}\textbf{\footnotesize{}72.39} & {\footnotesize{}78.06/}\textbf{\footnotesize{}78.57} & {\footnotesize{}45.71/}\textbf{\footnotesize{}45.96}\tabularnewline
        \hline 
      \end{tabular}
      \par
    \end{centering}

    \protect\caption{
      Performances (\%) of SE / OLE in the Yahoo Data. \emph{All results
      reported are averaged over self-defined three-fold.} All results with
      over 0.1 point improvement are marked.
    }

    \label{tb_full_table_all}
  \end{table*}

  \par

  \begin{table*}
    \begin{centering}
      \begin{tabular}{c|c|c|cccc}
        \multicolumn{1}{c}{} &  &  & \multicolumn{4}{c}{\textbf{\footnotesize{}Microsoft 10K}}\tabularnewline
        \hline 
        & \textbf{\footnotesize{}\#leaf} & \textbf{\footnotesize{}$\alpha$} & \textbf{\footnotesize{}NDCG@1} & \textbf{\footnotesize{}NDCG@3} & \textbf{\footnotesize{}NDCG@10} & \textbf{\footnotesize{}ERR}\tabularnewline
        \hline 
        &  & {\footnotesize{}0.06} & {\footnotesize{}47.43/47.14} & {\footnotesize{}46.14/}\textbf{\footnotesize{}46.46} & {\footnotesize{}48.60/}\textbf{\footnotesize{}49.17} & {\footnotesize{}35.90/}\textbf{\footnotesize{}36.13}\tabularnewline
        & {\footnotesize{}10} & {\footnotesize{}0.10} & {\footnotesize{}47.49/}\textbf{\footnotesize{}47.69} & {\footnotesize{}46.41/}\textbf{\footnotesize{}46.79} & {\footnotesize{}49.00/}\textbf{\footnotesize{}49.47} & {\footnotesize{}36.13/}\textbf{\footnotesize{}36.32}\tabularnewline
        \textbf{\footnotesize{}McRank} &  & {\footnotesize{}0.12} & {\footnotesize{}47.42/}\textbf{\footnotesize{}47.67} & {\footnotesize{}46.50/}\textbf{\footnotesize{}46.78} & {\footnotesize{}49.02/}\textbf{\footnotesize{}49.51} & {\footnotesize{}36.13/}\textbf{\footnotesize{}36.30}\tabularnewline
        \cline{2-7} 
        &  & {\footnotesize{}0.06} & {\footnotesize{}47.69/}\textbf{\footnotesize{}47.94} & {\footnotesize{}46.44/}\textbf{\footnotesize{}47.06} & {\footnotesize{}49.07/}\textbf{\footnotesize{}49.67} & {\footnotesize{}36.09/}\textbf{\footnotesize{}36.45}\tabularnewline
        & {\footnotesize{}20} & {\footnotesize{}0.10} & {\footnotesize{}47.52/}\textbf{\footnotesize{}48.04} & {\footnotesize{}46.76/}\textbf{\footnotesize{}47.24} & {\footnotesize{}49.36/}\textbf{\footnotesize{}49.80} & {\footnotesize{}36.26/}\textbf{\footnotesize{}36.54}\tabularnewline
        &  & {\footnotesize{}0.12} & {\footnotesize{}47.87/47.89} & {\footnotesize{}46.86/}\textbf{\footnotesize{}47.12} & {\footnotesize{}49.51/}\textbf{\footnotesize{}49.68} & {\footnotesize{}36.34/}\textbf{\footnotesize{}36.45}\tabularnewline
        \hline 
        &  & {\footnotesize{}0.06} & {\footnotesize{}47.42/}\textbf{\footnotesize{}47.98} & {\footnotesize{}46.21/}\textbf{\footnotesize{}46.54} & {\footnotesize{}48.22/}\textbf{\footnotesize{}48.70} & {\footnotesize{}36.44/}\textbf{\footnotesize{}36.68}\tabularnewline
        & {\footnotesize{}10} & {\footnotesize{}0.10} & {\footnotesize{}47.79/47.64} & {\footnotesize{}46.55/46.57} & {\footnotesize{}48.57/}\textbf{\footnotesize{}48.85} & {\footnotesize{}36.54/}\textbf{\footnotesize{}36.71}\tabularnewline
        \textbf{\footnotesize{}LambdaMART} &  & {\footnotesize{}0.12} & {\footnotesize{}47.45/}\textbf{\footnotesize{}47.79} & {\footnotesize{}46.32/}\textbf{\footnotesize{}46.62} & {\footnotesize{}48.54/}\textbf{\footnotesize{}48.94} & {\footnotesize{}36.42/}\textbf{\footnotesize{}36.67}\tabularnewline
        \cline{2-7} 
        &  & {\footnotesize{}0.06} & {\footnotesize{}48.01/}\textbf{\footnotesize{}48.19} & {\footnotesize{}46.52/}\textbf{\footnotesize{}46.87} & {\footnotesize{}48.78/}\textbf{\footnotesize{}49.13} & {\footnotesize{}36.74/}\textbf{\footnotesize{}36.88}\tabularnewline
        & {\footnotesize{}20} & {\footnotesize{}0.10} & {\footnotesize{}47.99/48.07} & {\footnotesize{}46.66/}\textbf{\footnotesize{}46.79} & \textbf{\footnotesize{}48.95/49.10} & {\footnotesize{}36.72/36.81}\tabularnewline
        &  & {\footnotesize{}0.12} & {\footnotesize{}47.63/47.67} & {\footnotesize{}46.69/46.51} & {\footnotesize{}49.02/49.03} & {\footnotesize{}36.70/36.61}\tabularnewline
        \hline 
      \end{tabular}
      \par
    \end{centering}

    \protect\caption{
      Performances (\%) of SE / OLE in the Microsoft 10K. \emph{All results
      reported are averaged over standard five-fold cross-validation respectively.}
      All results with over 0.1 point improvement are marked.
    }
    \label{tb_full_table_all_1}
  \end{table*}

  As (R)WSE has been shown to be a special case of OLE, we only compare
  SE and OLE in the scenario of learning to rank. We adopt two famous
  ranking systems with regression trees as weak learners. To be consistent 
  in implementation details, we used the same code template.
  Their differences are only from objective loss functions and regression
  tree fitting principles. 

  1. \emph{point-wise based McRank} (\cite{li2007mcrank}).
  The multi-class classification based system was reported to be strong
  in real-world datasets (\cite{wu2010adapting}), and is natural to be
  one of our baseline systems. 

  2. \emph{pair-wise based LambdaMART} (\cite{wu2010adapting}).
  This famous pair-wised system gained its reputation in Yahoo Challenge
  2010, as a combined system, mainly constructed on LambdaMART, winning
  the championship. In our work, we only compare with single LambdaMART
  systems, which are trained using NDCG loss. Maybe LambdaMART systems
  could be improved further using different configurations, but here
  it is not our concern.

  As shown by the proof, MART is an ideal intersection of these ideas,
  we do not use this system. For each system and algorithm,
  we set configurations as follows: the number of leaves is set as 10,
  20; the learning rate $\alpha$ in Eqn. \ref{eq_update} is set
  as 0.06, 0.1, 0.12. \emph{So there are six configurations for each
  system}. After examining the testing performance in the real-world datasets,
  we observed several hundreds of iterations (or regression trees) could almost lead to
  convergence, so we just set maximum number of iterations to 1000 for
  LambdaMART, and 2500 for McRank, as the latter converges more slowly. We report
  popular measures, \emph{NDCG@(1, 3, 10)} and \emph{ERR}. 

  \subsection{Experimental Comparison on Two Systems}

  Instead of providing limited testing results using the best parameter 
  from a validating data, we provide two kinds of testing results, one from
  converged training, anther from the whole training procedure.

  First, in Table \ref{tb_full_table_all} and \ref{tb_full_table_all_1}, we
  compare exact performances of SE and OLE for six configurations at
  predefined iteration. Empirically, training in large datasets, systems are easy
  to converge after sufficient iterations. Second, in Figure \ref{fig_impr_all}, we further
  provide a complete comparison in a whole training with
  difference configurations.

  \begin{figure*}[tp]
    \begin{centering}
      \includegraphics[scale=0.20]{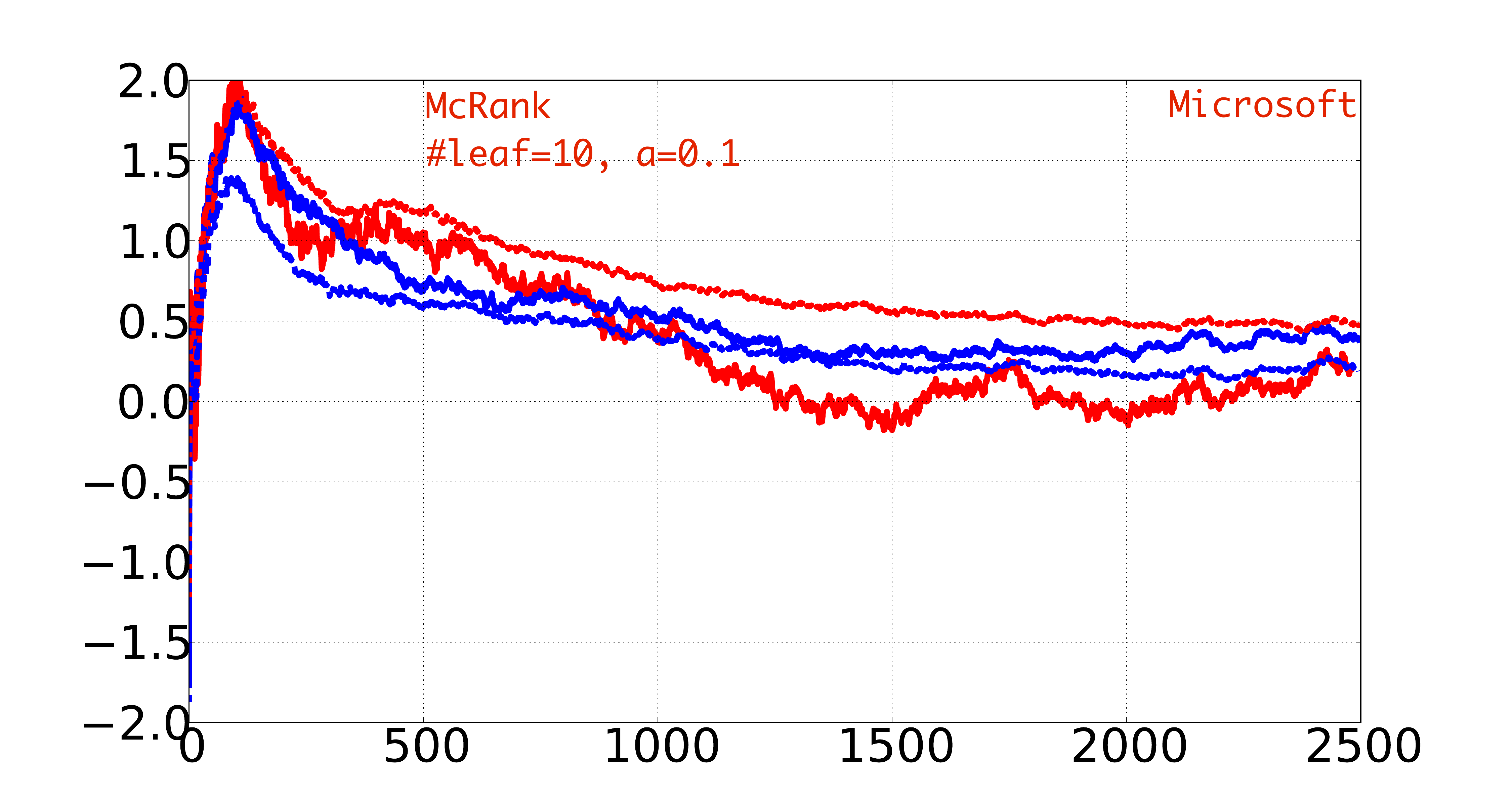}
      \includegraphics[scale=0.20]{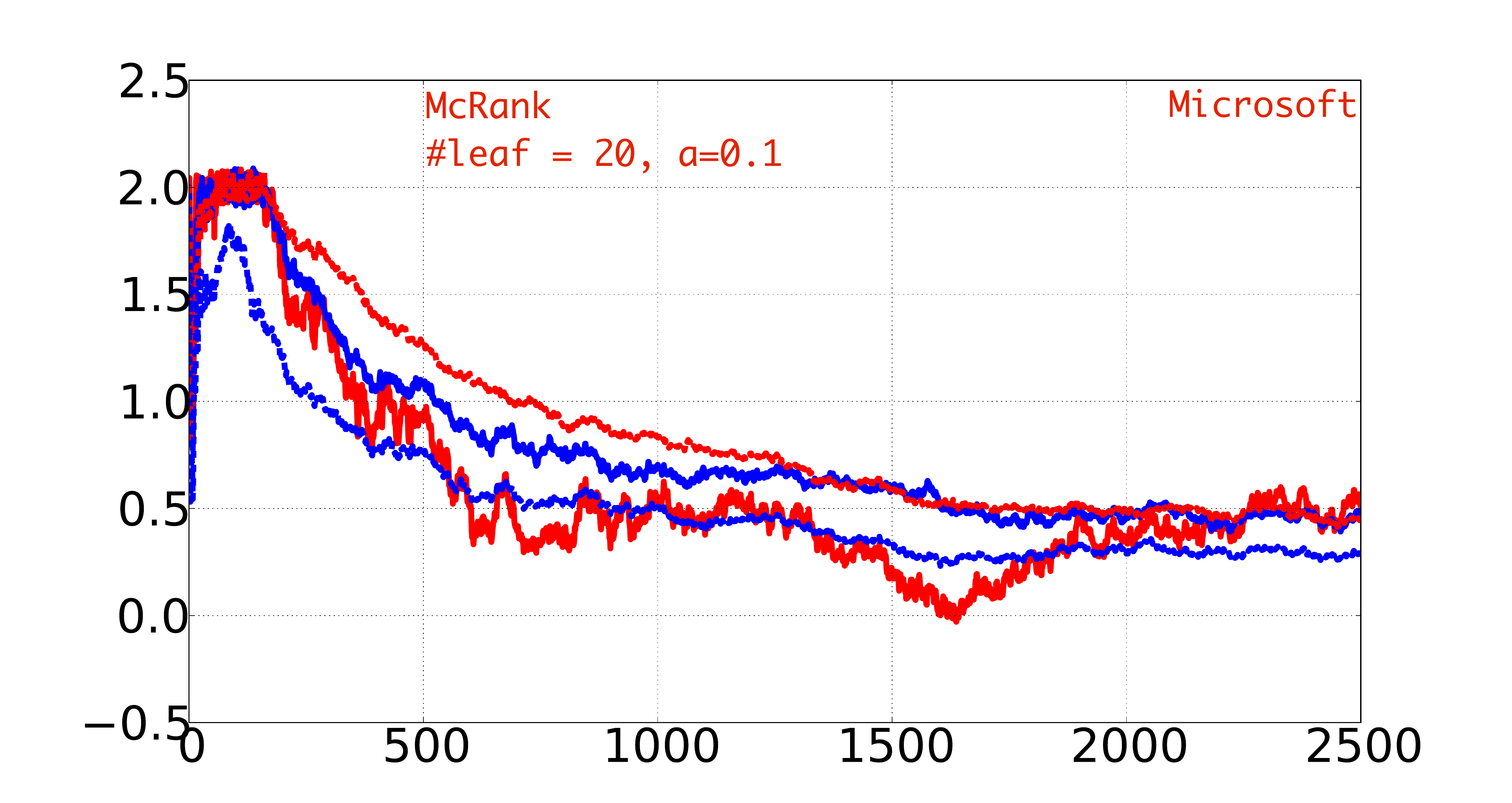}
      \par
      \includegraphics[scale=0.20]{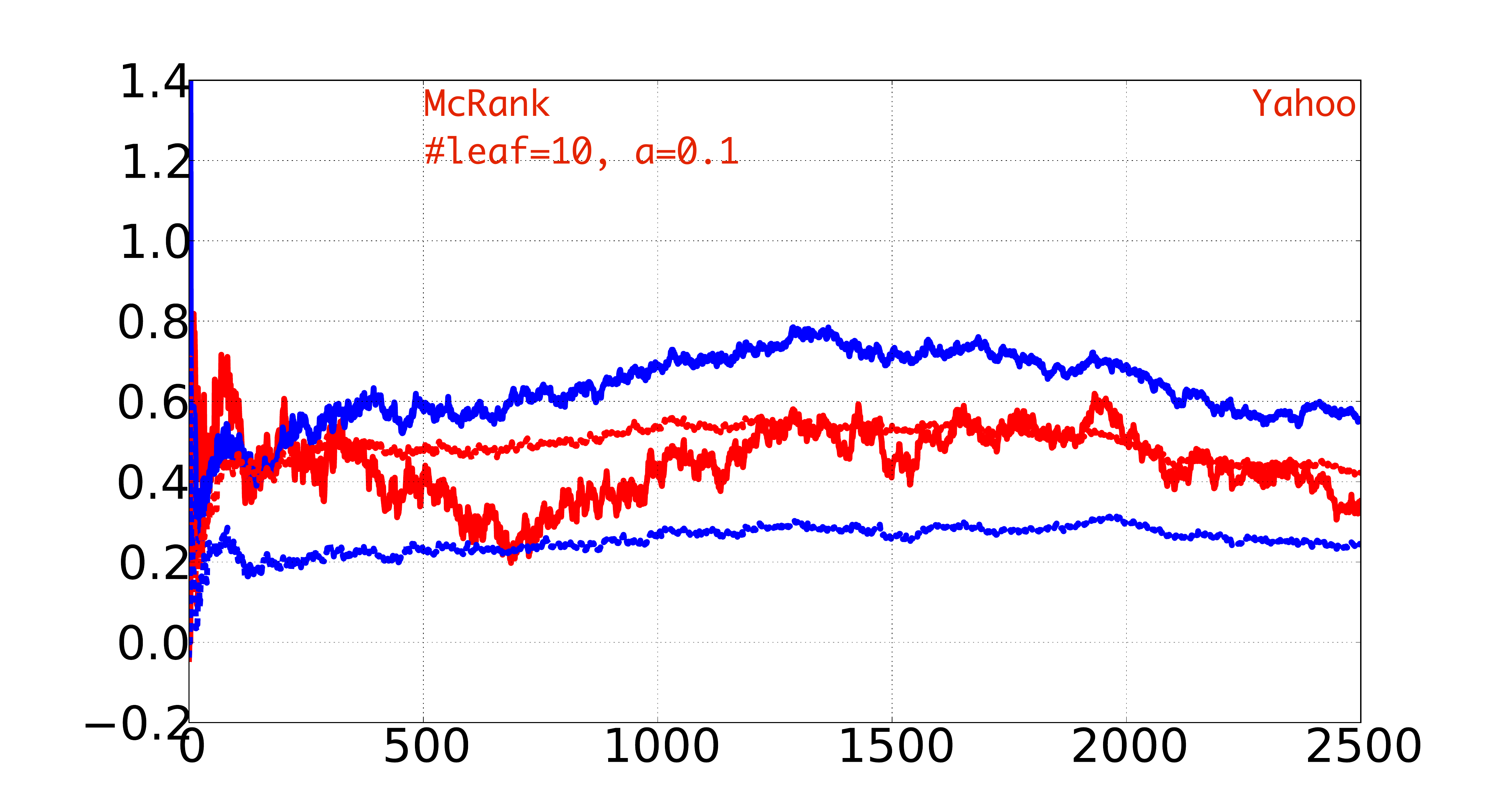}
      \includegraphics[scale=0.20]{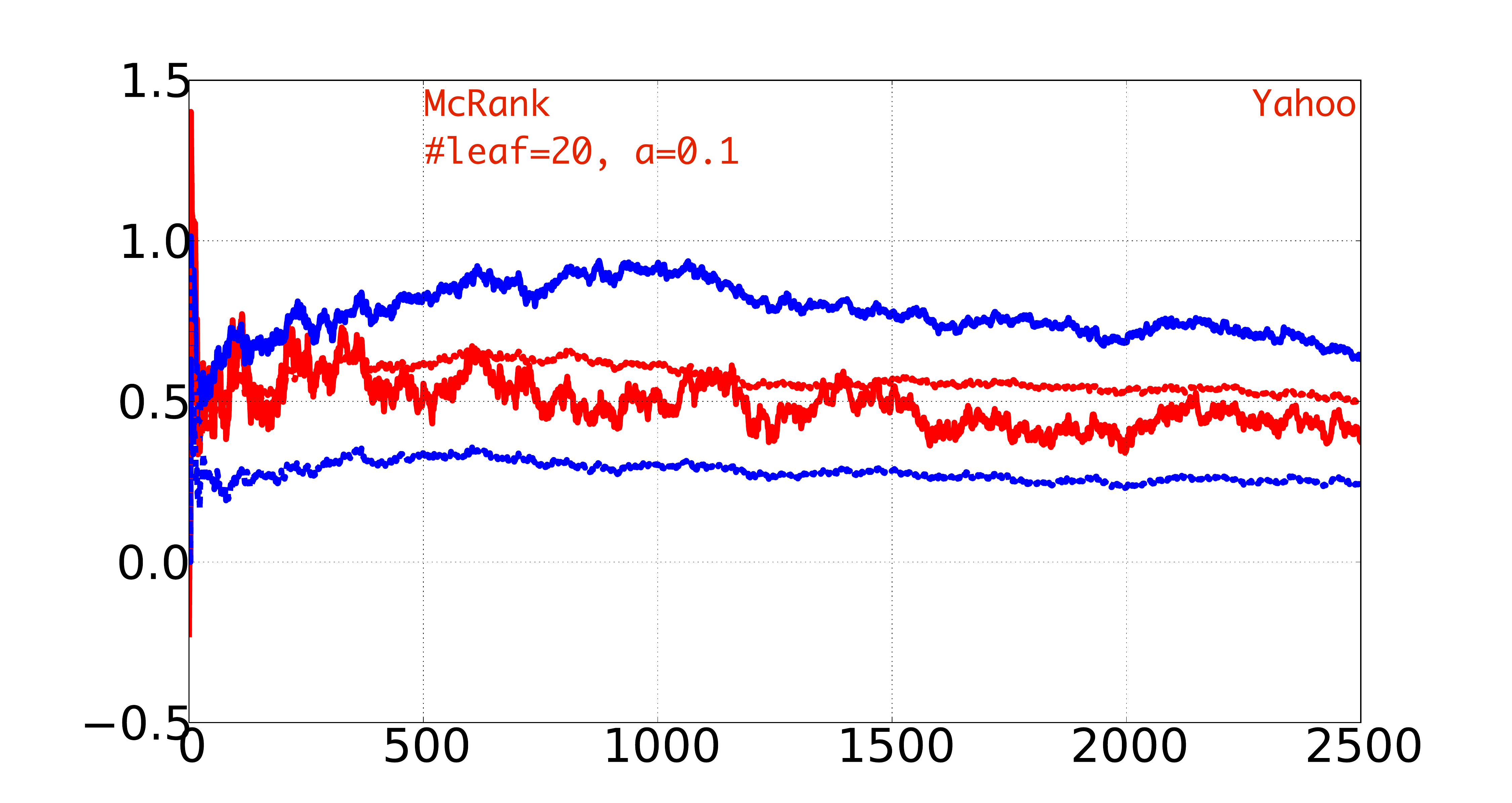}
      \par

      \includegraphics[scale=0.20]{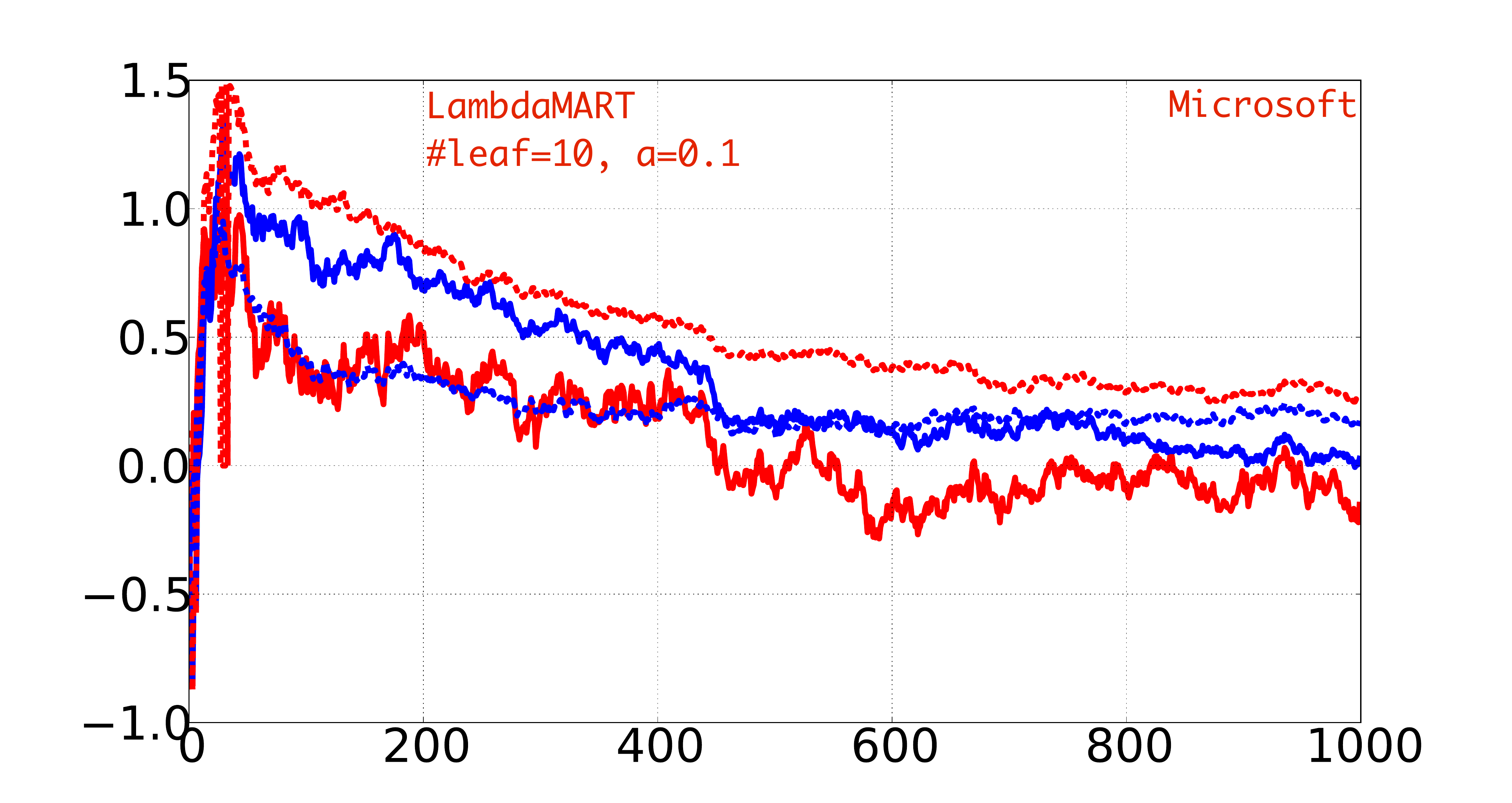}
      \includegraphics[scale=0.20]{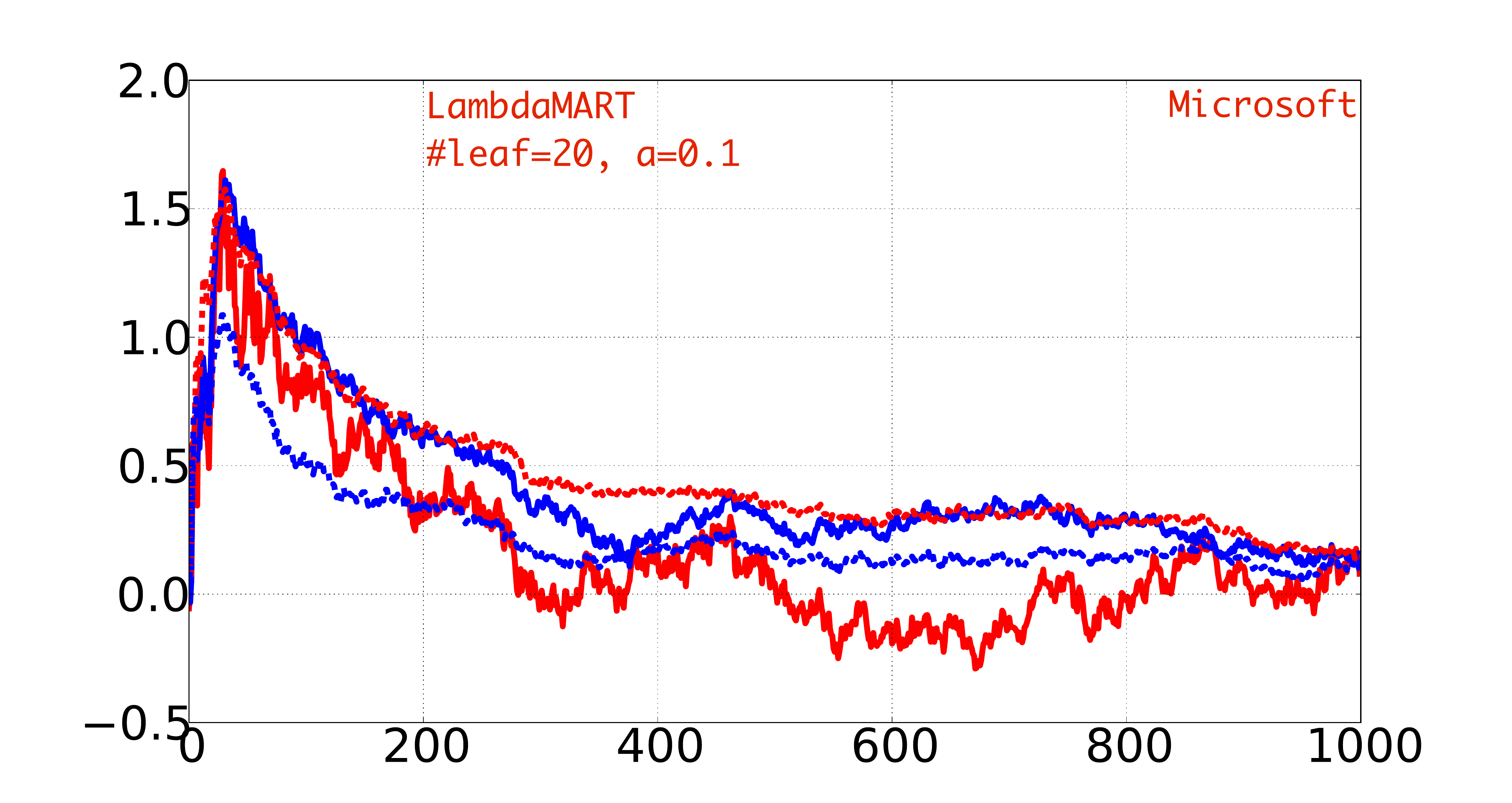}
      \par
      \includegraphics[scale=0.20]{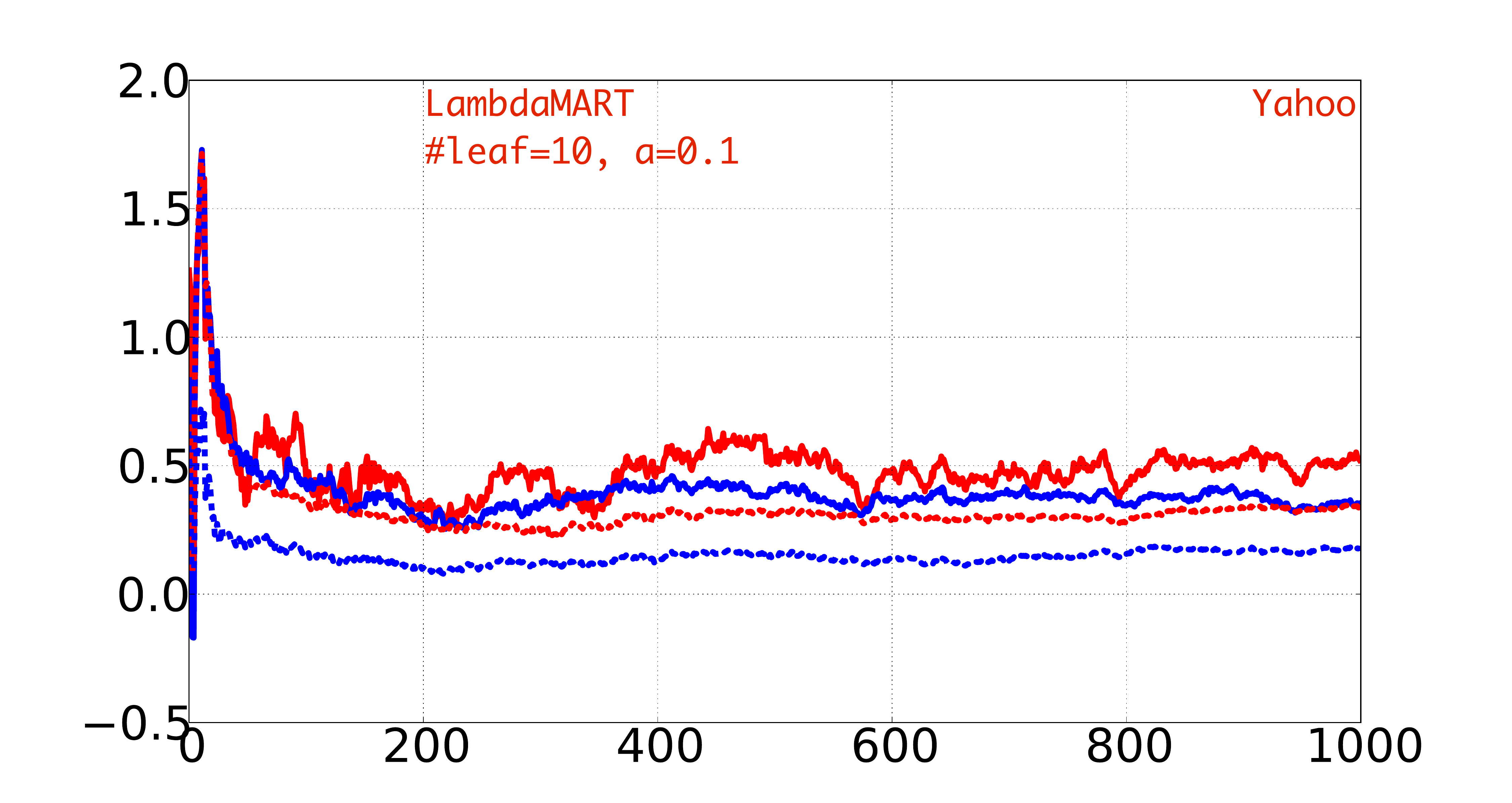}
      \includegraphics[scale=0.20]{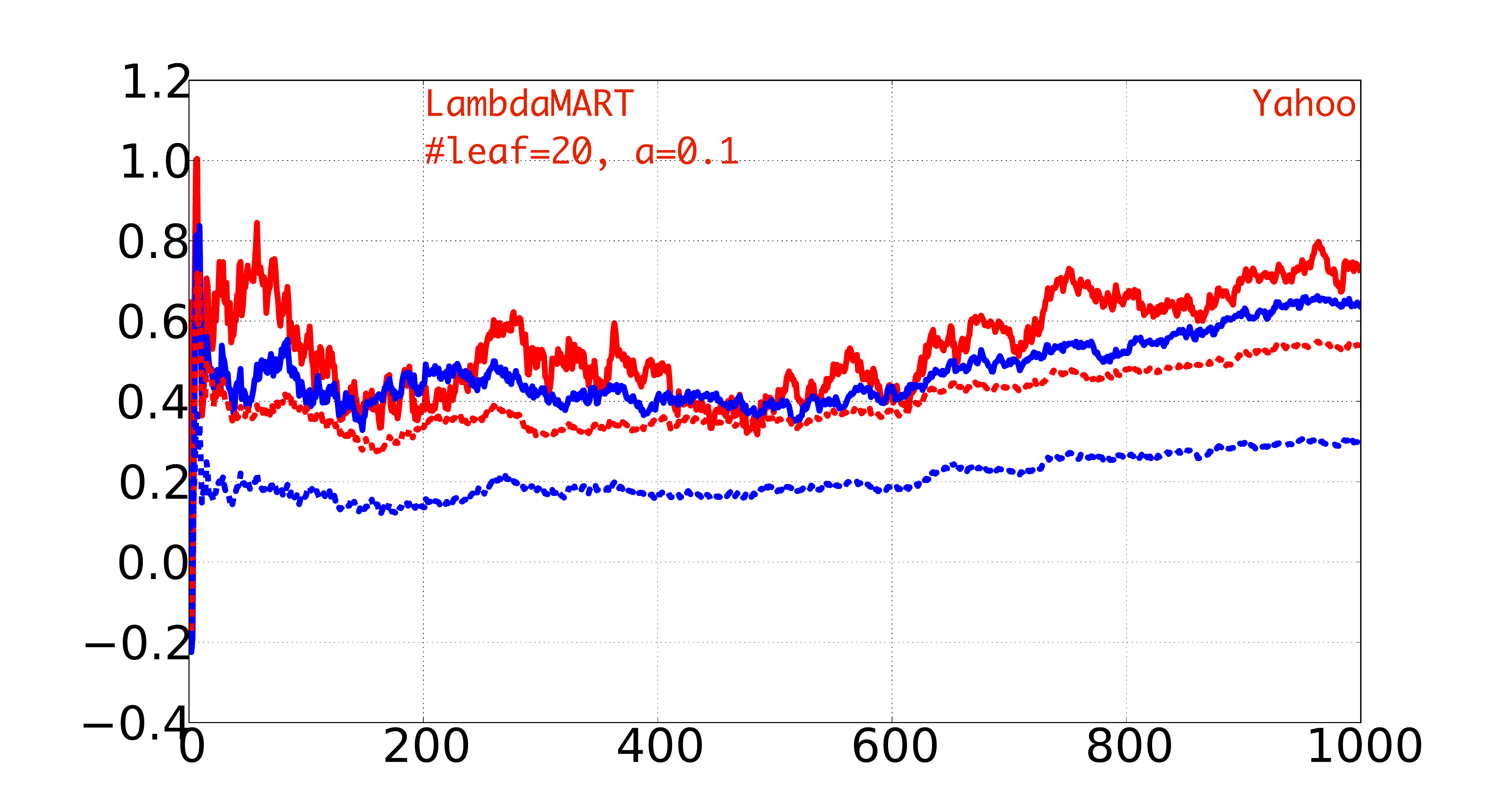}
      \par

      \includegraphics[scale=0.5]{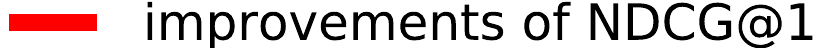}
      \includegraphics[scale=0.5]{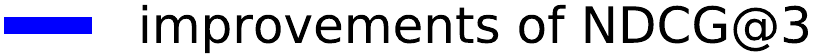}
      \includegraphics[scale=0.5]{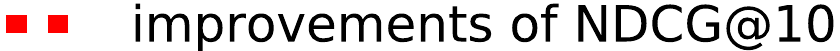}
      \includegraphics[scale=0.5]{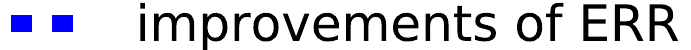}
      \par
    \end{centering}

    \protect\caption{
      Improvements (absolute \%) of OLE over SE when the learning rate is
      set as 0.06, 0.1, 0.12. 
      The X-axis denotes the iteration number, or the number of regression tree
      fitted.  
      The Y-axis denotes the difference that measure score from OLE minus that
      from SE for NDCG@(1, 3, 10) and ERR respectively.
      \emph{Each point in the figures has been averaged over a five-fold or three-fold cross-validation}. 
    }

    \label{fig_impr_all}
  \end{figure*}

  In table \ref{tb_full_table_all} and \ref{tb_full_table_all_1},
  for McRank and LambdaMART respectively, among 72 comparisons, OLE
  gains 68 and 56 improvements for at least over 0.1 point, and most
  of them are 0.3 to 0.4. These improvements are reasonable, as our
  baselines are strong, and in such large datasets. These statistics
  are based on six typical configurations, and demonstrate OLE is workable
  for the McRank and LambdaMART models in a general case. 

  We further analyze four measures separately, NDCG@(1, 3, 10) and ERR.
  Though both McRank (Figure \ref{fig_impr_all}, shown in the complementary
  material due to the space limit) and LambdaMART have been improved
  consistently with OLE, the NDCG@1 (real read line) and ERR (dotted
  blue line) have relatively smaller improvements. ERR is more difficult
  to improve than NDCG. 

  Improvements of NDCG@3 and NDCG@10 on McRank and LambdaMART are more
  robust in a variety of configurations. As NDCG@1 is computed on the
  first document predicted by models, and ERR is computed on the whole
  of ranking documents whose numbers are usually several dozens, in
  practice, the first page with 10 links returned by a search engine
  are more desired by users. So we think it may be more useful to improve
  NDCG@3 and NDCG@10 measures.

  As OLE is supposed to have a faster convergence than SE, we also have
  a statistics of objective losses in the final iteration. OLE indeed
  leads to smaller objective losses, but not by that much, about 0.32\%
  - 1\%. As this work only focuses on the splitting rule in a single
  node, we also tried different strategies to generate node. Width-first
  search and depth-first search. Interestingly, depth-first search runs
  poorly for both baselines and our method. This is an open question
  and left to future exploration. We thus adopted the width-first search
  and limits the number of leaves. 

  Regarding the running time, there is no loss for the systems with
  derivative additive objective losses compared to SE in gradient boosting.
  Typical such systems are point-wise based. But for pair-wise and list-wise
  based, OLE suffers from extra overheads of maintaining exact second
  derivatives of objective loss function. In an incremental updating
  style, this overhead is about 30\% of SE for pair-wise based, and
  regarding list-wise, there may be more.

  \section{Conclusion}

  In this paper, we propose a minimum objective loss based tree construction
  algorithm in the boosting framework, and analyze two existent tree
  construction principles, least square error and (robust) weighted square error. The
  former is widely used in the gradient boosting and practical learning to rank
  systems, while the latter is famous in LogitBoost and classification area.
  We successful build a relationship between our method and WSE in LogitBoost.
  We prove that WSE is just a special case of our method. This provides a theoretical support
  for (robust) LogitBoost and point-wise based ranking systems. Based
  on our analysis, we show MART is an ideal connection to SE, WSE, and
  OLE, and obtain a more concise formula for MART. Finally, for a
  full empirical comparison of the three principles, we implement
  two strong ranking systems, and examine them with a variety
  of configurations of regression trees in two largest public datasets.
  Our results indicate that our proposed method is better used
  for McRank, LambdaMART and MART systems.

  \bibliographystyle{plain}
  \nocite{*}
  \bibliography{references}
\end{document}